\documentclass[journal]{IEEEtran}
\IEEEoverridecommandlockouts
\usepackage{bbding}
\usepackage{cite}
\usepackage{amsmath,amssymb,amsfonts}
\usepackage{amsthm}
\usepackage{algorithmic}
\usepackage{algorithm}
\usepackage{graphicx}
\usepackage{textcomp}
\usepackage{xcolor}
\usepackage{booktabs}
\usepackage{multirow}
\usepackage{makecell}
\usepackage{url}
\newtheorem{theorem}{Theorem}
\newtheorem{proposition}{Proposition}
\newtheorem{corollary}{Corollary}
\newtheorem{assumption}{Assumption}
\newtheorem{remark}{Remark}

\floatname{algorithm}{Algorithm}

\def\BibTeX{{\rm B\kern-.05em{\sc i\kern-.025em b}\kern-.08em
    T\kern-.1667em\lower.7ex\hbox{E}\kern-.125emX}}

\usepackage{tikz}
\usepackage{hyperref}
\newcommand{\eb}{\bar{\varepsilon}}
\newcommand{\pmess}{p_m^{\mathrm{ESS}}}
\definecolor{lime}{HTML}{A6CE39}
\DeclareRobustCommand{\orcidicon}{
	\begin{tikzpicture}
		\draw[lime, fill=lime] (0,0)
		circle[radius=0.16]
		node[white]{{\fontfamily{qag}\selectfont \tiny \.{I}D}}; 
	\end{tikzpicture}
	\hspace{-2mm}
}
\foreach \x in {A, ..., Z}{%
	\expandafter\xdef\csname orcid\x\endcsname{\noexpand\href{https://orcid.org/\csname orcidauthor\x\endcsname}{\noexpand\orcidicon}}
}

\begin{document}

\title{Securing Cooperative Sensing in UAV Swarms Against Conformity-Driven Byzantine Attacks}

	\author{Ruixing~Ren\hspace{-1.5mm}\orcidA{}, Junhui~Zhao\hspace{-1.5mm}\orcidB{},~\IEEEmembership{Senior~Member,~IEEE},  Qiuping~Li\hspace{-1.5mm}\orcidD{}, He~Fang\hspace{-1.5mm}\orcidC{},~\IEEEmembership{Member,~IEEE},\\ Jiamin~Li\hspace{-1.5mm}\orcidE{},~\IEEEmembership{Member,~IEEE}, and Dongming~Wang\hspace{-1.5mm}\orcidG{},~\IEEEmembership{Member,~IEEE},
	\thanks{This work was supported in part by National Natural Science Foundation of China under Grant U25B2007, in part by the the Fundamental Research Funds for the central Universities under Grant 2025JBZX060, in part by National Engineering Research Center of System Technology for High-Speed	Railway and Urban Rail Transit under Grant 2024YJ255. (Corresponding author: Junhui Zhao.)}%
	\thanks{Ruixing Ren, Junhui Zhao are with the School of Electronic and Information Engineering, Beijing Jiaotong University, Beijing 100044, China. (e-mail: renruixing0604@163.com; junhuizhao@hotmail.com)
		
	Qiuping Li is with the National Computer Network Emergency Response Technical Team/Coordination Center of China (CNCERT/CC), Beijing 100029, China (e-mail: qiupingli\_bj@163.com).
		
	He Fang is with the College of Computer and Cyber Security, Fujian Normal University, Fuzhou 200240, China (e-mail: fanghe@fjnu.edu.cn).
		
	Jiamin Li, and Dongming Wang are with the National Mobile Communications Research Laboratory, Southeast University, Nanjing 210096, China, and also with Purple Mountain Laboratories, Nanjing 211111, China (e-mail: jiaminli@seu.edu.cn; wangdm@seu.edu.cn).
		
	}
	
}

\maketitle

\begin{abstract}
In integrated sensing and communication (ISAC)-enabled 6G unmanned aerial vehicle (UAV) swarm networks, the widely adopted imitation-based conformity cooperation mechanism can be exploited by Byzantine attackers to fabricate false consensus, causing the effective error probability of normal UAVs to evolve dynamically and far exceed their inherent sensing errors, which invalidates conventional fusion methods built on the independence assumption. This paper proposes a conformity-aware Byzantine-resilient fusion framework that couples evolutionary game theory with maximum a posteriori (MAP) estimation. First, the strategy updates of normal UAVs are characterized by bounded-rational opinion dynamics, and the evolution dynamics of the misinformation ratio together with its evolutionarily stable state (ESS) are derived under death-birth updating. Three theoretical results are then established: under heterogeneous per-node sensing errors, the zeroth-order ESS depends on the error distribution only through its mean; a closed-form first-order weak-selection correction to the ESS is obtained, together with an exact mean-field fixed point valid for arbitrary selection intensity; and it is revealed that swarm-level misinformation can overwhelm the majority if and only if the attack probability exceeds one half, with this threshold independent of both the sensing error and the malicious ratio. Embedding the predicted error dynamics into a per-node MAP rule, the resulting fusion mechanism achieves nearly 100\% situation-inference accuracy under different network topologies, attack intensities, network scales, and sensing-error distributions, and maintains accuracy above 99\% under $\pm20\%$ parameter mismatch. In contrast, majority voting, reputation weighting, and independent fusion collapse completely once the majority-flip threshold is crossed.
\end{abstract}


\section{Introduction}
Low-altitude airspace is rapidly evolving into a densely interconnected economic and operational space, in which UAVs undertake reconnaissance, monitoring, delivery, and communication-relay missions \cite{RenPC,AirGuard}. IMT-2030 lists resilience, trustworthiness, and ultra-reliable low-latency communication as first-tier design objectives for 6G \cite{ITU2030}. Integrated sensing and communication (ISAC), as one of its flagship enabling technologies, allows UAVs to exchange data and simultaneously sense the environment over the same 6G link \cite{ISACSurvey,Yao,RenTVTIoV}; mission-critical applications such as emergency response and intelligent transportation impose stringent reliability and latency requirements on the resulting situation inference \cite{RenUAV,RenITS,RenIOTJ}.

A single UAV is constrained by endurance, sensing coverage, and fault tolerance \cite{RenUAV}. Swarms commonly compensate for these limitations through \emph{local conformity}, whereby each UAV iteratively aligns its decisions with those of its neighbors: nodes repeatedly exchange and weight one-bit decisions, which can provably approach centralized detection performance \cite{QCons,CSS}, and such mechanisms have been embedded in many cooperative sensing protocols.

The realization of conformity, however, is not unique. Narrow-sense quantized consensus \cite{QCons} is centered on unbiased aggregation, in which normal nodes perform no strategy selection; whereas in bio-inspired imitation-based coordination \cite{Vicsek,MDCons,DB,OhtsukiNature} and in some consensus-enhanced protocols, normal nodes iteratively switch among candidate behaviors using the fitness of their neighbors as a reference. This paper focuses on the latter, namely imitation-based conformity; the robustness of the former is determined by aggregation weights and voting thresholds and lies outside the scope of this paper.

Imitation-based conformity likewise faces Byzantine threats. Attackers can forge data or sabotage the consensus process \cite{DFalse,PrivCons,SwarmSpoof}, and in ISAC-6G swarms, where sensing and communication share the same open wireless medium \cite{Yao}, this exposure is further amplified. Classical models assume that normal nodes err independently with a fixed probability while attackers actively flip their reports, so that the fusion center can suppress attacks via statistical anomaly detection, majority voting, or reputation weighting \cite{WSPRT,BetaCSS}.

These defenses overlook a crucial fact: under conformity rules, the reporting errors of normal UAVs are neither fixed nor independent. When normal nodes tend to stay consistent with their neighbors, malicious UAVs can fabricate false consensus and lure them, round after round, into adopting the wrong strategy, so that the effective error probability \emph{evolves dynamically} and may far exceed the inherent sensing error. A fusion center that still assumes a constant error probability suffers severe model mismatch, directly jeopardizing the high-reliability and low-latency objectives of ISAC-enabled 6G swarms.

Evolutionary game theory (EGT) provides a natural language for describing such phenomena: boundedly rational individuals repeatedly adjust their strategies through local interactions. Death-birth (DB) updating on graphs has been systematically studied in the evolution of cooperation \cite{DB,OhtsukiNature,NowakFive} and has been extended to model opinion and rumor spreading in social networks \cite{LinPhD,Centola}, but it has not yet been used to analyze the coupling between imitation-based conformity and Byzantine attacks in ISAC-enabled UAV swarms. Taking the Ohtsuki--Nowak framework \cite{DB} as the starting point and following the modeling philosophy of \cite{LinPhD}, this paper constructs an analysis and fusion methodology for cooperative sensing in UAV swarms, and develops a set of new theoretical and algorithmic results. The main contributions are as follows.
\begin{itemize}
	\item Using graph evolutionary game theory to model the process by which normal UAVs are progressively misled by malicious neighbors, deriving the evolution dynamics and the ESS of the misinformation ratio, and establishing the majority-flip threshold for independent attackers of fixed intensity.
	\item Proving, under a mean-field neighborhood-statistics assumption, that with per-node heterogeneous sensing errors the zeroth-order ESS depends on the error distribution only through its mean, providing a basis for mean-field fusion design under realistic channel conditions.
	\item Deriving a first-order closed-form correction to the transition probability together with the corresponding ESS shift, giving an exact mean-field fixed point valid for arbitrary selection intensity, and revealing the structural amplification effect on graphs.
	\item Embedding the predicted error dynamics into a per-node MAP rule to track online the true error probability of normal reports; nearly 100\% accuracy is achieved under different topologies, attack intensities, scales, and error distributions, and accuracy remains above 99\% under $\pm20\%$ parameter mismatch.
\end{itemize}

The remainder of this paper is organized as follows. Section \ref{Sec2} reviews related work; Section \ref{Sec3} presents the system architecture and the evolutionary-game formulation; Section \ref{Sec4} derives the evolution dynamics and the ESS; Section \ref{Sec5} develops the theoretical extensions; Section \ref{Sec6} designs the conformity-aware MAP fusion mechanism; Section \ref{Sec7} reports the simulation results; and Section \ref{Sec8} concludes the paper.

\section{Related Work}\label{Sec2}
This section reviews related work along four threads: UAV swarm conformity control, consensus-based cooperative sensing and distributed detection, Byzantine-resilient decision fusion, and applications of evolutionary game theory to swarm security. In ISAC-enabled 6G networks, sensing and communication share wireless resources to support highly reliable, low-latency services \cite{RenITS,RenIOTJ}, and their security and resilience have become key issues for IMT-2030 \cite{ITU2030}. However, existing research mostly focuses on waveform design, resource allocation, and physical-layer reliability, with little attention paid to swarm-level error propagation induced by the malicious exploitation of local conformity rules in cooperative sensing. This paper takes precisely this network-centric security perspective.

The Vicsek model pioneered the description of individuals adjusting their motion according to the average direction of their neighbors \cite{Vicsek}, spawning a large body of conformity-control research; the introduction of molecular-dynamics potentials further improved the efficiency and stability of swarm conformity \cite{MDCons}. To cope with non-ideal communication, improved consensus protocols, adaptive algorithms, and event-triggered schemes still enable effective conformity in UAV swarms \cite{Kucherov,HanCons}; at the mission level, conformity theory has also been applied to swarm encirclement without preset formations \cite{WenSwarm} and to robust trajectory prediction under saturation attacks \cite{ShangPred}. All of these works assume that the nodes are cooperative and honest, and do not consider the security risk that malicious nodes may exploit the conformity rule to mislead normal nodes.

It should be clarified that the quantized consensus studied in \cite{QCons,CSS} is centered on unbiased aggregation, in which normal nodes act as passive information weighters; this paper instead models imitation updates at the strategy level. The two mechanisms are fundamentally different: the conclusions of this paper directly characterize the latter, while the former lies outside its coverage.

Byzantine attacks allow an adversary to arbitrarily tamper with the transmitted data, and in severe cases can prevent consensus from converging at all \cite{BlockchainTrust}. Position-spoofing attacks can be detected using neighborhood-proximity features \cite{SpoofDetect}, but such methods treat the attack as exogenous and assume that normal nodes are always honest. Recent work has designed a variety of Byzantine-resilient consensus protocols: digital-twin hierarchical defenses for heterogeneous swarms \cite{CuiResilient}; digital-twin and human-in-the-loop hierarchical control \cite{GongHIL}; two-layer security protocols under composite attacks \cite{GongTwoLayer}; event-triggered control for hybrid air--sea swarms under DoS \cite{LiAirSea}; zero-sum-game self-triggered control under intermittent communication/DoS \cite{WuZeroSum}; and privacy-preserving consensus under DoS \cite{PrivCons}. All of these methods assume that normal nodes are either fully honest or err with a fixed independent probability, and none models the dynamic evolution of reporting behavior driven by local conformity. When swarm-level error propagation occurs, a fusion center assuming a fixed error probability suffers severe mismatch. Reputation-based defenses represented by majority voting and reputation weighting \cite{WSPRT,BetaCSS} further update reputations using the global decision as the reference; once the majority is misled en masse, the reputation mechanism itself becomes contaminated and keeps rewarding misleading nodes.

EGT effectively characterizes the process by which individuals progressively adjust their strategies through local interactions, and has been applied to rumor spreading and malicious-information diffusion in social networks \cite{LinPhD,Centola}. The distributed decision-making pattern of UAV swarms fits naturally with the local-interaction premise of EGT. When normal UAVs tend to stay consistent with their neighbors' reports, malicious UAVs can fabricate false consensus and lure them away from the honest strategy. This process is not a one-shot random disturbance but a multi-round iterative strategy evolution, and quantitative modeling of this process and of its impact on fusion decisions has long been lacking. The graph evolutionary game framework used in this paper follows the bounded-rational opinion-security model of \cite{LinPhD}, redeploys it in ISAC-enabled UAV swarms, and extends it with results on heterogeneous robustness, weak selection, and fusion.

\begin{figure}[t]
	\centerline{\includegraphics[width=3.4in,keepaspectratio]{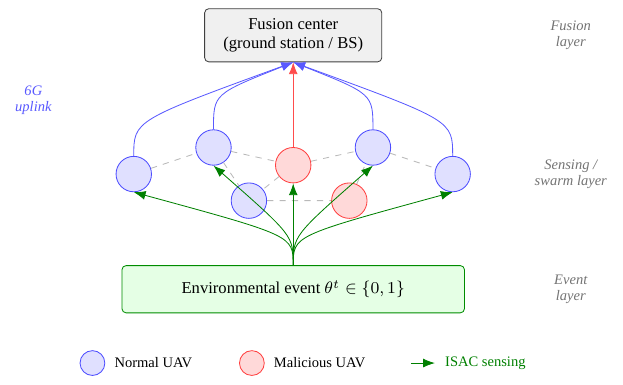}}
	\caption{Network architecture of ISAC-enabled 6G UAV swarm cooperative sensing.}
	\label{fig0}
\end{figure}

\section{System Architecture and Game Formulation}\label{Sec3}
The considered network is illustrated in Fig. \ref{fig0} and comprises a sensing layer, a 6G communication layer, and a fusion layer. The communication topology is an undirected connected graph $G=(\mathcal{V},\mathcal{E})$ formed by $N$ UAVs. At each discrete time $t$, every UAV senses a binary event state $\theta^t\in\{0,1\}$ (e.g., target present/absent, spectrum occupied/idle) using ISAC signals. Due to the limited echo SINR, it errs with probability $\varepsilon$, and its local decision $u_i^t$ satisfies $P(u_i^t\ne\theta^t)=\varepsilon$. After sensing, each UAV sends its report $r_i^t$ to the fusion center over the 6G uplink, while neighboring nodes simultaneously exchange reports through local interactions, constituting the conformity coupling inherent to the swarm; the fusion center aggregates all reports and infers the true situation.

In open airspace, some UAVs may be compromised into malicious nodes executing an attack strategy: besides forging observations and tampering with reports, more importantly they can exploit the local conformity rule followed by normal UAVs to drag them, round after round, away from their true perceptions, so that the effective error probability of normal reports evolves dynamically and far exceeds the inherent sensing error $\varepsilon$, threatening the highly reliable, low-latency situation inference promised by ISAC-6G networks.

EGT characterizes the process by which individuals iteratively adjust their strategies through local interactions, and fits the neighbor-based distributed decision-making pattern of UAV swarms. Following \cite{DB,LinPhD}, this paper models the interaction between a normal UAV and its neighbors (including malicious nodes) as an evolutionary game on a graph.

\textbf{Strategy space.} A normal UAV has two strategies: the normal strategy $S_n$, reporting the sensed value honestly, $r_i^t=u_i^t$; and the misleading strategy $S_m$, reporting the opposite value, $r_i^t=\bar{u}_i^t$. Adopting $S_m$ is not active malice but a wrong decision induced by neighbors (especially malicious nodes) under the conformity rule. Malicious UAVs adopt the attack strategy $S_a$, flipping their reports with probability $P_a$. Let $p_m$ denote the fraction of normal UAVs adopting $S_m$, and let $\beta$ denote the ratio of malicious UAVs to normal UAVs.

\textbf{Task payoffs.} When interacting with neighbors, a UAV earns task payoffs depending on whether the two parties' reports agree. A central UAV adopting $S_n$ earns $u_{ns}$ from each neighbor whose previous report agreed with its own, and $u_{nd}$ otherwise; when adopting $S_m$ it earns $u_{ms}$ and $u_{md}$, respectively. Cooperative tasks bias UAVs toward agreement, hence $u_{ns}>u_{ms}$ and $u_{nd}>u_{md}$; moreover, agreement yields higher payoffs than disagreement, i.e., $u_{ns}>u_{nd}$ and $u_{ms}>u_{md}$.

\textbf{Fitness.} Following \cite{LinPhD}, the fitness of a player is
\begin{equation}\label{eq:fitness}
	\pi=(1-\alpha)B+\alpha\cdot U,
\end{equation}
where $B$ is the baseline fitness capturing inherent attributes; individuals less susceptible to others have a larger $B$. Assuming a homogeneous network, all UAVs take $B=1$. The weak-selection coefficient $\alpha$ captures the conformity tendency of the UAVs and is generally very small. $U$ aggregates the total payoff from interactions with neighbors.

Since a UAV cannot know the true strategies of its neighbors for certain, it evaluates them using its own perception as the reference: a neighbor whose report agrees with its own perception is regarded as adopting $S_n$, and a disagreeing one as adopting $S_m$. This is consistent with the limited-information condition of distributed sensing.

\textbf{Strategy update rule.} Normal UAVs continuously adjust their strategies under the influence of their neighbors. Following \cite{LinPhD}, this paper adopts death-birth (DB) updating:
\begin{itemize}
	\item \emph{Death}: in each round, a randomly selected normal UAV abandons its current strategy;
	\item \emph{Birth}: it imitates the strategy of a neighbor with probability proportional to that neighbor's fitness.
\end{itemize}

\textbf{Evolutionarily stable state (ESS).} The ESS is the stable equilibrium of the evolutionary process: once reached, small behavioral perturbations are damped out and the system returns to equilibrium. In this paper, the system reaches the ESS when the misinformation ratio $p_m^t$ no longer changes, i.e.,
\begin{equation}
	\pmess:\quad \hat{p}_m^t = 0,
\end{equation}
where $\hat{p}_m^t$ denotes the time derivative of $p_m^t$.

\section{Evolutionary Dynamics and Stable State of the UAV Swarm}\label{Sec4}
Based on the setup of Section \ref{Sec3}, this section derives the evolutionary dynamics by which normal UAVs are progressively misled by malicious UAVs into erroneous strategies and tend toward a stable state; it further gives the exact mean-field transition probability valid for arbitrary selection intensity.

\subsection{State Transition Probability}\label{Sec4A}
Among the $k$ neighbors connected to a central UAV, suppose that $k_m$ neighbors choose to adopt strategy $S_m$\footnote{For malicious UAVs, this corresponds to maliciously flipped reports.} and have modified their outcomes accordingly, while $k_n=k-k_m$ neighbors choose to adopt strategy $S_n$. Conditional on an observation error probability $\varepsilon$, let $k^S_n$ denote the number of neighbors whose reports agree with the true system state, and $k^S_m=k-k^S_n$ the number whose reports disagree. It follows that
\begin{equation}\label{eq:kmS}
	k^S_m=(1-\varepsilon)k_m+\varepsilon k_n, k^S_n=(1-\varepsilon)k_n+\varepsilon k_m.
\end{equation}

In practice, a UAV also cannot determine whether its own observation contains an error. Therefore, the fitness function must be analyzed separately for adopting $S_n$ or $S_m$ when the observation is correct or erroneous. When the central UAV has no observation error, adopting the misleading strategy $S_m$, i.e., $r_i^t=\bar u_i^t\ne\theta^t$, yields the fitness
\begin{equation}\label{eq:piA}
	\pi_m^{(A)}=(1-\alpha)+\frac{\alpha}{k}\left[ k^S_m u_{ms}+ (k-k^S_m)u_{md}\right],
\end{equation}
while adopting the normal strategy $S_n$, i.e., $r_i^t=u_i^t=\theta^t$, yields
\begin{equation}\label{eq:piB}
	\pi_n^{(A)}=(1-\alpha)+\frac{\alpha}{k}\left[ k^S_n u_{ns}+ (k-k^S_n)u_{nd}\right].
\end{equation}
Similarly, when the central UAV has an observation error, the fitness functions are
\begin{equation}\label{eq:piC}
	\pi_m^{(B)}=(1-\alpha)+\frac{\alpha}{k}\left[ k^S_n u_{ms}+ (k-k^S_n)u_{md}\right],
\end{equation}
\begin{equation}\label{eq:piD}
	\pi_n^{(B)}=(1-\alpha)+\frac{\alpha}{k}\left[ k^S_m u_{ns}+ (k-k^S_m)u_{nd}\right],
\end{equation}
where the interaction terms take the average-payoff form, consistent with the fitness definition in \eqref{eq:fitness}.

According to the DB update rule, let $P_{n\to m}$ denote the probability that a UAV switches from the truthful strategy to the misleading strategy, whose numerator is the sum of the fitnesses of neighbors judged to adopt the strategy opposite to the true system state, and whose denominator is the total fitness of all neighbors, denoted by $\Omega$. When the central UAV has no observation error, this can be written as
\begin{equation}\label{eq:PA}
	P^{(A)}_{n\to m}=\frac{k_m(1-\varepsilon)\pi^{(A)}_m+k_n\varepsilon\pi_n^{(B)}}{\Omega},
\end{equation}
where
\begin{equation}\label{eq:Omega}
	\begin{aligned}
		\Omega&=\left[k_n(1-\varepsilon)\pi^{(A)}_n+k_m\varepsilon\pi_m^{(B)}\right]\\
		&\quad+\left[k_m(1-\varepsilon)\pi^{(A)}_m+k_n\varepsilon\pi_n^{(B)}\right].
	\end{aligned}
\end{equation}
When the central UAV has an observation error, it becomes
\begin{equation}\label{eq:PB}
	P^{(B)}_{n\to m}=\frac{k_n(1-\varepsilon)\pi^{(A)}_n+k_m\varepsilon\pi_m^{(B)}}{\Omega}.
\end{equation}
Combining \eqref{eq:PA} and \eqref{eq:PB} and taking the expectation over the central observation error, the probability that an arbitrary normal UAV in the network changes from the truthful strategy to the misleading strategy is
\begin{equation}\label{eq:Pnm}
	P_{n\to m}=(1-\varepsilon)P^{(A)}_{n\to m}+\varepsilon P^{(B)}_{n\to m}.
\end{equation}

To further analyze \eqref{eq:Pnm}, substitute \eqref{eq:piA}--\eqref{eq:piD} into \eqref{eq:PA}--\eqref{eq:PB} and then into \eqref{eq:Pnm}; after simplification one obtains
\begin{equation}\label{eq:Ppi}
	P_{n\to m}=\frac{\bar\varepsilon^2 k_m\pi_m^{(A)}+\varepsilon^2 k_m\pi_m^{(B)}+\varepsilon\bar\varepsilon k_n\pi_n^{(B)}+\varepsilon\bar\varepsilon k_n\pi_n^{(A)}}{\bar\varepsilon k_n\pi_n^{(A)}+\varepsilon k_n\pi_n^{(B)}+\bar\varepsilon k_m\pi_m^{(A)}+\varepsilon k_m\pi_m^{(B)}},
\end{equation}
where $\bar\varepsilon=1-\varepsilon$. Note that \eqref{eq:piA}--\eqref{eq:piD} all contain the common term $1-\alpha$ followed by a payoff-dependent term, denoted by $\Psi_i^{(X)}$. They can thus be written uniformly as
\begin{equation}\label{eq:piPsi}
	\pi_i^{(X)}=(1-\alpha)+\alpha\cdot\Psi_i^{(X)},
\end{equation}
and substituting this into the numerator of \eqref{eq:Ppi} gives
\begin{equation}
	\begin{aligned}
		&\bar\varepsilon^2 k_m\left[(1-\alpha)+\alpha\Psi_m^{(A)}\right]
		+\varepsilon^2 k_m\left[(1-\alpha)+\alpha\Psi_m^{(B)}\right]\\
		&\quad+\varepsilon\bar\varepsilon k_n\left[(1-\alpha)+\alpha\Psi_n^{(B)}\right]
		+\varepsilon\bar\varepsilon k_n\left[(1-\alpha)+\alpha\Psi_n^{(A)}\right].
	\end{aligned}
\end{equation}
Separating the $(1-\alpha)$ terms from the $\alpha$ terms and collecting like terms yields
\begin{equation}
	(1-\alpha)\left(\bar\varepsilon^2 k_m+\varepsilon^2 k_m+2\varepsilon\bar\varepsilon k_n\right)+\alpha\cdot\Phi_1,
\end{equation}
where
\begin{equation}\label{eq:Phi1}
	\Phi_1=\bar\varepsilon^2 k_m\Psi_m^{(A)}+\varepsilon^2 k_m\Psi_m^{(B)}+\varepsilon\bar\varepsilon k_n\Psi_n^{(B)}+\varepsilon\bar\varepsilon k_n\Psi_n^{(A)}
\end{equation}
is a composite term depending on the payoff parameters. Similarly, applying the same operations to the denominator finally gives
\begin{equation}
	(1-\alpha)\left(\bar\varepsilon k_n+\varepsilon k_n+\bar\varepsilon k_m+\varepsilon k_m\right)+\alpha\cdot\Phi_2,
\end{equation}
\begin{equation}\label{eq:Phi2}
	\Phi_2=\bar\varepsilon k_n\Psi_n^{(A)}+\varepsilon k_n\Psi_n^{(B)}+\bar\varepsilon k_m\Psi_m^{(A)}+\varepsilon k_m\Psi_m^{(B)}.
\end{equation}
Combining the two, \eqref{eq:Ppi} can be consolidated as
\begin{equation}\label{eq:Pexpanded}
	P_{n\to m}=\frac{(1-\alpha)\left(\bar\varepsilon^2 k_m+\varepsilon^2 k_m+2\varepsilon\bar\varepsilon k_n\right)+\alpha\Phi_1}{(1-\alpha)k+\alpha\Phi_2},
\end{equation}
where in the denominator $\bar\varepsilon k_n+\varepsilon k_n+\bar\varepsilon k_m+\varepsilon k_m=(\bar\varepsilon+\varepsilon)(k_n+k_m)=1\cdot k=k$.

Since $\alpha$ is very small, the following approximate expansion holds:
\begin{equation}\label{eq:fracexp}
	\frac{1+\mu\alpha}{1+\nu\alpha}=1+(\mu-\nu)\alpha+O(\alpha^2).
\end{equation}
Dividing both the numerator and the denominator of \eqref{eq:Pexpanded} by $(1-\alpha)k$ puts the denominator into the form $1+\nu\alpha$:
\begin{equation}
	P_{n\to m}=\frac{\frac{\bar\varepsilon^2 k_m+\varepsilon^2 k_m+2\varepsilon\bar\varepsilon k_n}{k}+\frac{\alpha}{(1-\alpha)k}\Phi_1}{1+\frac{\alpha\Phi_2}{(1-\alpha)k}}.
\end{equation}
Because $\alpha$ is small, $\frac{1}{1-\alpha}\approx1+O(\alpha)$, so the numerator can be written compactly as
\begin{equation}
	\underbrace{\frac{(\bar\varepsilon^2+\varepsilon^2)k_m+2\varepsilon\bar\varepsilon k_n}{k}}_{A_0}
	+\underbrace{\frac{\Phi_1}{k}}_{\mu}\cdot\alpha+O(\alpha^2),
\end{equation}
and the denominator as
\begin{equation}
	1+\frac{\alpha\Phi_2}{(1-\alpha)k}\approx1+\underbrace{\frac{\Phi_2}{k}}_{\nu}\cdot\alpha+O(\alpha^2).
\end{equation}
Therefore,
\begin{equation}
	\begin{aligned}
		P_{n\to m}&=\frac{A_0+\mu\alpha+O(\alpha^2)}{1+\nu\alpha+O(\alpha^2)}
		=\frac{A_0\left(1+\frac{\mu}{A_0}\alpha\right)}{1+\nu\alpha}+O(\alpha^2)\\
		&=A_0\left[1+\left(\frac{\Phi_1}{A_0 k}-\frac{\Phi_2}{k}\right)\alpha\right]+O(\alpha^2)\\
		&=A_0+\alpha\left(\frac{\Phi_1}{k}-A_0\cdot\frac{\Phi_2}{k}\right)+O(\alpha^2).
	\end{aligned}
\end{equation}
Using $k_n=k-k_m$ and $\bar\varepsilon=1-\varepsilon$, the first term simplifies to give
\begin{equation}\label{eq:P0}
	P_{n\to m}=2\varepsilon(1-\varepsilon)+(1-2\varepsilon)^2\frac{k_m}{k}+\alpha\cdot\varpi(k_m)+O(\alpha^2),
\end{equation}
where
\begin{equation}\label{eq:varpi}
	\varpi(k_m)=\frac{\Phi_1}{k}-\frac{(\bar\varepsilon^2+\varepsilon^2)k_m+2\varepsilon\bar\varepsilon k_n}{k}\cdot\frac{\Phi_2}{k}
\end{equation}
collects the payoff-dependent first-order terms. Equation \eqref{eq:P0} is exact in $\alpha$ up to first order: its zeroth-order term is used to derive the population evolution equation in the next subsection, and the closed form of the first-order term is given in Section \ref{Sec5}.

\subsection{Population Evolution Equation}\label{Sec4B}
Recalling the definitions above, under the large-scale network assumption the fraction of normal UAVs adopting the misleading strategy is $p_m$, and the fraction of malicious neighbors adopting the attack strategy is $P_a$. For an arbitrary normal UAV, its $k$ neighbors include $k_o$ ordinary neighbors, so the number of malicious neighbors is $k_b=\beta k_o$. Hence the expected total number of neighbors adopting the misleading strategy (or the attack strategy, both of which manifest as flipped reports) is
\begin{equation}
	\mathbb{E}[k_m]=\mathbb{E}[k_{o,m}]+\mathbb{E}[k_{b,m}]=k_o p_m+k_b P_a=k_o(p_m+\beta P_a).
\end{equation}
The expected total number of neighbors is $\mathbb{E}[k]=k_o(1+\beta)$, so
\begin{equation}\label{eq:rho}
	\mathbb{E}\!\left[\frac{k_m}{k}\right]=\frac{p_m+\beta P_a}{1+\beta}=:\rho_s.
\end{equation}

According to the DB update rule, the rate of change of the misleading fraction $p_m$ satisfies \cite{DB}
\begin{equation}\label{eq:dpdt}
	\hat{p}_m=\mathbb{E}\left[P_{n\to m}\right]-p_m.
\end{equation}
Substituting \eqref{eq:P0} into \eqref{eq:dpdt} and taking the expectation over the neighbor distribution gives
\begin{equation}
	\begin{aligned}
		\hat{p}_m&=\mathbb{E}\!\left[2\varepsilon(1-\varepsilon)+(1-2\varepsilon)^2\frac{k_m}{k}+\alpha\cdot\varpi(k_m)\right]\\
		&\quad-p_m+O(\alpha^2).
	\end{aligned}
\end{equation}
Substituting \eqref{eq:rho} yields
\begin{equation}
	\begin{aligned}
		\hat{p}_m&=2\varepsilon(1-\varepsilon)+(1-2\varepsilon)^2\frac{p_m+\beta P_a}{1+\beta}-p_m\\
		&\quad+\alpha\,\mathbb{E}[\varpi(k_m)]+O(\alpha^2).
	\end{aligned}
\end{equation}

Since $\alpha\ll1$ is the weak-selection coefficient and $\mathbb{E}[\varpi(k_m)]$ is a bounded constant determined by the model parameters, $\alpha\mathbb{E}[\varpi(k_m)]$ is a higher-order small quantity. In the limit $\alpha\to0$ this term vanishes, and the main contribution to the system evolution comes from the zeroth-order terms. In the subsequent analysis this paper adopts the zeroth-order approximation, i.e., ignoring $\alpha$ and retaining only the $O(1)$ part. This gives the zeroth-order evolution equation
\begin{equation}\label{eq:zeroth}
	\begin{aligned}
		\hat{p}_m&\approx 2\varepsilon(1-\varepsilon)+(1-2\varepsilon)^2\frac{p_m+\beta P_a}{1+\beta}-p_m\\
		&=-p_m\left[1-\frac{(1-2\varepsilon)^2}{1+\beta}\right]+2\varepsilon(1-\varepsilon)+\frac{(1-2\varepsilon)^2\beta P_a}{1+\beta}.
	\end{aligned}
\end{equation}
This is a linear differential equation in $p_m$, which can be written in the standard linear form
\begin{equation}
	\hat{p}_m=Ap_m+B,
\end{equation}
\begin{equation}\label{A}
	A=-1+\frac{(1-2\varepsilon)^2}{1+\beta},
\end{equation}
\begin{equation}\label{B}
	B=2\varepsilon(1-\varepsilon)+\frac{(1-2\varepsilon)^2\beta P_a}{1+\beta}.
\end{equation}
The equilibrium $\pmess$ is obtained from $\hat{p}_m=0$ as $\pmess=-B/A$, which expands to
\begin{equation}\label{eq:ess}
	\pmess=\frac{2\varepsilon(1-\varepsilon)(1+\beta)+(1-2\varepsilon)^2\beta P_a}{(1+\beta)-(1-2\varepsilon)^2}.
\end{equation}

Rewriting the equation in terms of the error variable $\varrho=p_m-\pmess$ and substituting gives
\begin{equation}
	\begin{aligned}
		\frac{d\varrho}{dt}&=\hat{p}_m=Ap_m+B=A(\varrho+\pmess)+B\\
		&=A\varrho+A(-B/A)+B=A\varrho.
	\end{aligned}
\end{equation}
The solution of this equation is $\varrho(t)=\varrho(0)e^{At}$. Hence the stability of the equilibrium is entirely determined by the sign of the coefficient $A$. In practical UAV swarm scenarios, the sensing error generally satisfies $0<\varepsilon\le0.5$ and the malicious ratio satisfies $\beta>0$, in which case $A<0$. The exponential term $e^{At}$ decays with time, $\varrho(t)\to0$, the error keeps shrinking, and the system eventually returns to the equilibrium, which is therefore stable. Thus the equilibrium $\pmess$ under the zeroth-order approximation is asymptotically stable.

\subsection{Majority-Flip Threshold}
A natural question then arises: in this steady state, are most normal UAVs still honest? If the majority has already been misled, then the classical strategy of trusting the majority fails at the fusion center. It is therefore necessary to compare $\pmess$ with $1/2$.

\begin{corollary}[Majority-flip threshold]\label{cor:flip}
	For $0<\varepsilon<1/2$ and $\beta>0$,
	\begin{equation}
		\pmess\gtrless\frac{1}{2}\quad\Longleftrightarrow\quad P_a\gtrless\frac{1}{2}.
	\end{equation}
	In particular, this threshold is independent of both $\varepsilon$ and $\beta$.
\end{corollary}
\begin{proof}
	Let $c=(1-2\varepsilon)^2\in(0,1)$. Since $2\varepsilon(1-\varepsilon)=(1-c)/2$,
	\begin{align}
		\pmess-\frac12
		&=\frac{\frac{1-c}{2}(1+\beta)+c\beta P_a-\frac12\left[(1+\beta)-c\right]}{(1+\beta)-c}\notag\\
		&=\frac{c\beta\left(P_a-\frac12\right)}{(1+\beta)-c}.
	\end{align}
	Since $(1+\beta)-c>0$ and $c\beta>0$, the sign is determined by $(P_a-\tfrac12)$.
\end{proof}

Corollary \ref{cor:flip} shows that when $P_a>1/2$, the aggregate reports of the swarm are systematically inverted, and any defense that trusts the majority is doomed to fail; when $P_a<1/2$, the majority remains trustworthy and classical defenses survive. It should be pointed out that the above threshold relies on the DB imitation update and the payoff ordering of Section \ref{Sec3}; if other update rules are adopted, the threshold must be re-derived, although the failure boundary of majority voting can be analyzed in a similar manner.

\subsection{Exact Mean-Field Transition Probability}
The derivation of $P_{n\to m}$ above performed a weak-selection expansion in $\alpha$, yielding the first-order approximation in \eqref{eq:P0}. In practical systems $\alpha$ need not always be negligible, so it is necessary to retain the full fitness in order to verify the reliability of the approximation. Keeping the forms of the fitnesses $\pi_m^{(A)}$, $\pi_n^{(A)}$, $\pi_m^{(B)}$, $\pi_n^{(B)}$ from Section \ref{Sec4A} unchanged, merely denoting the central error by $\varepsilon_f$ and the neighbor error by $\varepsilon_n$, and making no weak-selection truncation, one directly obtains the exact transition probability:
\begin{equation}\label{eq:Pexact} P^{\mathrm{ex}}_{n\to m}(\alpha,\varepsilon_f,\varepsilon_n,\rho) = (1-\varepsilon_f)\frac{\mathrm{Num}_A}{\Omega} + \varepsilon_f\frac{\mathrm{Num}_B}{\Omega}, \end{equation}
where $\rho=k_m/k$ is the local misinformation ratio, and
\begin{align} 
\mathrm{Num}_A &= k_m(1-\varepsilon_n)\pi_m^{(A)} + k_n\varepsilon_n\pi_n^{(B)},\\
\mathrm{Num}_B &= k_n(1-\varepsilon_n)\pi_n^{(A)} + k_m\varepsilon_n\pi_m^{(B)},\\
\Omega &= \mathrm{Num}_A + \mathrm{Num}_B. 
\end{align}
When $\alpha=0$, \eqref{eq:Pexact} reduces exactly to the zeroth-order term of \eqref{eq:P0}; its Taylor coefficient at $\alpha=0$ is precisely the closed-form function $G$ derived in Proposition \ref{prop:firstorder} below. The exact mean-field ESS $p^{\mathrm{mf}}_m(\alpha)$ is the solution of the fixed-point equation $p=P^{\mathrm{ex}}_{n\to m}(\alpha,\varepsilon,\varepsilon,\rho_s(p))$ and can be computed by simple iteration. This object serves in Section \ref{Sec7} as the all-order theoretical benchmark.

\section{Model Refinement and Robustness Analysis}\label{Sec5}
This section relaxes assumptions and improves accuracy on top of the previous section. It first considers per-node heterogeneous sensing errors and examines the dependence of the ESS on the error distribution; it then gives the first-order closed-form correction under weak selection; finally, it discusses graph-structural effects beyond the mean-field approximation.

\subsection{Heterogeneous Sensing Errors}
In real ISAC swarms, UAVs experience different echo SINRs due to differences in distance, viewing angle, and occlusion, so the per-node decision errors are heterogeneous. Concretely, consider the binary hypothesis test $H_0$: $\theta=0$ versus $H_1$: $\theta=1$. Under additive white Gaussian noise, the matched-filter output has signal-to-noise ratio $\gamma_i$; with maximum-likelihood decision, the false-alarm and miss probabilities are equal and given by
\begin{equation}\label{eq:epsmap}
	\varepsilon_i=P(\mathrm{Judgment-error}|\gamma_i)=Q\!\left(\sqrt{2\gamma_i}\right),
\end{equation}
where $Q(x)=\frac{1}{\sqrt{2\pi}}\int_x^\infty e^{-t^2/2}\,dt$ is the complementary cumulative distribution function of the standard normal distribution. Let $\eb=N^{-1}\sum_i\varepsilon_i$ denote the population-average error.

For a normal UAV $i$ with error $\varepsilon_i$, average neighbor error $\eb$, and local misinformation ratio $\rho_i$, the zeroth-order transition probability generalizes \eqref{eq:P0} to
\begin{equation}\label{eq:Pi}
	P_i=\varepsilon_i+\eb-2\varepsilon_i\eb+(1-2\varepsilon_i)(1-2\eb)\rho_i.
\end{equation}
Indeed, conditioning on the central observation as in \eqref{eq:Pnm} and letting $\alpha\to0$ gives $P_i=(1-\varepsilon_i)[\eb+(1-2\eb)\rho_i]+\varepsilon_i[1-\eb-(1-2\eb)\rho_i]$, which simplifies to \eqref{eq:Pi}.

\begin{assumption}[Mean-field neighborhood statistics]\label{ass:mf}
	The local misinformation ratio $\rho_i$ is uncorrelated with $\varepsilon_i$ in the population-average sense, i.e.,
	\begin{equation}
		\frac1N\sum_i(1-2\varepsilon_i)\rho_i = (1-2\eb)\,\frac1N\sum_i\rho_i.
	\end{equation}
\end{assumption}
This assumption has clear physical plausibility. $\varepsilon_i$ is determined by physical-layer channel conditions (distance, viewing angle, occlusion), while $\rho_i$ is determined by network-layer topology and neighbor behavior, and the two are assigned independently in the swarm. The sensing SINR of each UAV depends only on its own geometry, while neighbor relations are determined by the communication radius or a preconfigured topology, and the two are mutually uncorrelated.

\begin{theorem}[Zeroth-order robustness to heterogeneity]\label{thm:hetero}
	Under Assumption \ref{ass:mf}, the population-average zeroth-order dynamics of the misinformation ratio under heterogeneous errors $\{\varepsilon_i\}$ is
	\begin{equation}\label{eq:heterodyn}
		\hat{p}_m=2\eb(1-\eb)+(1-2\eb)^2\frac{p_m+\beta P_a}{1+\beta}-p_m,
	\end{equation}
	identical to the homogeneous dynamics \eqref{eq:zeroth} with $\varepsilon=\eb$. Consequently, the ESS depends on the distribution of $\{\varepsilon_i\}$ only through the mean $\eb$.
\end{theorem}
\begin{proof}
	Averaging \eqref{eq:Pi} over the normal UAVs and applying Assumption \ref{ass:mf} together with $\rho_i\to\rho_s$ gives
	\begin{align}
		\frac1N\sum_i P_i
		&=\underbrace{\eb+\eb-2\eb\cdot\eb}_{=\,2\eb(1-\eb)}
		+\underbrace{\frac1N\sum_i(1-2\varepsilon_i)}_{=\,1-2\eb}(1-2\eb)\rho_s\\
		&=2\eb(1-\eb)+(1-2\eb)^2\rho_s.
	\end{align}
	Substituting $\rho_s=(p_m+\beta P_a)/(1+\beta)$ yields \eqref{eq:heterodyn}.
\end{proof}

\begin{remark}
	Theorem \ref{thm:hetero} shows that the fusion center only needs the population mean $\eb$ to predict the evolution trend of normal reports, which greatly simplifies online implementation; the per-node errors $\varepsilon_i$ still matter for the fusion likelihood itself, and Section \ref{Sec6} exploits them through the per-node MAP rule. It should be noted that the above robustness is premised on Assumption \ref{ass:mf} (uncorrelated $\varepsilon_i$ and $\rho_i$); if channel conditions are strongly coupled with network topology (e.g., edge nodes simultaneously have high $\varepsilon_i$ and high $\rho_i$), the conclusion of the theorem requires correction, and such cases are left for future work.
\end{remark}

\subsection{First-Order Weak-Selection Closed-Form Correction}

The first-order term $\varpi(k_m)$ in \eqref{eq:P0} contains the payoff information, but its form is $\Phi_1/k - P_0\cdot\Phi_2/k$, where $\Phi_1$ and $\Phi_2$ are each weighted sums of four $\Psi_i^{(X)}$, and each $\Psi_i^{(X)}$ is in turn a linear combination of payoffs. This three-layer nesting makes it impossible to directly read off the relation between the sign of $\varpi$ and the payoff parameters. This section simplifies it, under the mean-field approximation, into a closed form linear in $\{u_{ns},u_{nd},u_{ms},u_{md}\}$, thereby analytically characterizing the directional influence of payoff-driven selection on the ESS.

The starting point of the derivation is the exact mean-field transition probability \eqref{eq:Pexact}. Viewing it as a function of $\alpha$ and taking a first-order Taylor expansion at $\alpha=0$:
\begin{equation}
	P^{\mathrm{ex}}_{n\to m}(\alpha) = P^{\mathrm{ex}}_{n\to m}(0) + \frac{\partial P^{\mathrm{ex}}_{n\to m}}{\partial\alpha}\Big|_{\alpha=0}\cdot\alpha + O(\alpha^2).
\end{equation}
The constant term $P^{\mathrm{ex}}_{n\to m}(0)$ is precisely the zeroth-order term of \eqref{eq:P0}, and the first-order coefficient is the mean-field form of $\varpi$. Under the mean-field approximation, setting $k_m/k=\rho$, applying the quotient rule to \eqref{eq:Pexact}, and simplifying with computer algebra (SymPy symbolic expansion, cross-checked against a central-difference scheme to a relative error below $10^{-10}$), one obtains the following closed-form result.

\begin{proposition}[First-order correction term]\label{prop:firstorder}
	Let $e=\eb$ and $\rho=\rho_s$. Define the sensing-reliability measure $c=(1-2e)^2\in(0,1)$ and the auxiliary quantities
	\begin{equation}
		q = 2e\rho - e - \rho,\quad
		Y = 4e^2\rho^2 - 4e^2\rho + e^2 - 4e\rho^2 + 4e\rho - e + \rho^2.
	\end{equation}
	Then the first-order coefficient $\mathbb{E}[\varpi(k_m)]$ of $\alpha$ in the population transition probability is
	\begin{equation}\label{eq:G}
		G(e,\rho) = g_{ns}\,u_{ns} + g_{nd}\,u_{nd} + g_{ms}\,u_{ms} + g_{md}\,u_{md},
	\end{equation}
	where
	\begin{align}
		g_{ns} &= -c(\rho-1)\,q(q+1),\\
		g_{nd} &= c(\rho-1)\,Y,\\
		g_{ms} &= -\rho\,c\,q(q+1),\\
		g_{md} &= \rho\,c\,(Y-2\rho+1).
	\end{align}
	The first-order population dynamics is therefore
	\begin{equation}\label{eq:firstorder}
		\hat{p}_m = Ap_m + B + \alpha\,G(\eb,\rho_s) + O(\alpha^2),
	\end{equation}
	where $A$ and $B$ are given in (\ref{A}) and (\ref{B}). Setting $\hat{p}_m=0$ and applying the implicit function theorem around the zeroth-order ESS $p_0^*$ yields the ESS shift
	\begin{equation}\label{eq:ess1}
		p_m^*(\alpha) = p_0^* - \alpha\,\frac{G(\eb,\rho_0^*)}{A} + O(\alpha^2),
	\end{equation}
	where $\rho_0^* = (p_0^* + \beta P_a)/(1+\beta)$.
\end{proposition}

It should be pointed out that the above analysis is based on the mean-field approximation, in which every node sees the population-average $\rho_s$. On finite graphs, fluctuations of the local $\rho_i$ and the correlation between node strategies and neighbor composition generally rescale the effective selection intensity \cite{DB,OhtsukiNature,NowakFive}. Section \ref{Sec7} simulates two microscopic dynamics on random regular graphs (the proposed DB transition and direct fitness imitation), and the results show that both agree with the mean-field prediction; the deviation manifests as a slightly steeper descent caused by structural amplification, which lies within an acceptable range.

\section{Conformity-Aware MAP Fusion}\label{Sec6}
The ground station receives all reports $\mathbf{r}=\{r_i^t\}$ and must infer the true state sequence $\Theta=(\theta^1,\dots,\theta^T)$. Under a uniform prior, the MAP estimate is
\begin{equation}
	\theta^*=\arg\max_{\Theta}P(\mathbf{r}\mid\Theta).
\end{equation}
Introduce the type vector $\xi=(\xi_1,\dots,\xi_N)$, where $\xi_i=1$ marks a malicious UAV. By the law of total probability,
\begin{equation}\label{eq:map}
	\theta^*=\arg\max_{\Theta}\sum_{\xi}P(\mathbf{r}\mid\Theta)P(\xi).
\end{equation}

\textbf{1) Malicious report model.} A malicious UAV flips its report with probability $P_a$; accounting for the sensing error, its report error probability is
\begin{equation}\label{eq:delta}
	\delta=\varepsilon(1-P_a)+(1-\varepsilon)P_a,
\end{equation}
so $P(r_i^t\mid\theta^t,\xi_i=1)=(1-\delta)^{q_i(t)}\delta^{1-q_i(t)}$, where $q_i(t)=\mathbb{1}[r_i^t=\theta^t]$.

\textbf{2) Normal report model.} The report of a normal UAV depends on its perception and, through conformity, on the evolutionary state. At $t=1$ it reports honestly with error $\varepsilon$. For $t>1$, letting the misinformation ratio be $p_m^t$, its report error probability is
\begin{equation}\label{eq:gamma}
	\gamma^t=(1-\varepsilon)p_m^t+\varepsilon(1-p_m^t),
\end{equation}
so $P(r_i^t\mid\mathbf{r}^{t-1},\theta^t,\xi_i=0)=(1-\gamma^t)^{q_i(t)}(\gamma^t)^{1-q_i(t)}$. Defining $\gamma^1=\varepsilon$ unifies the two cases.

Given the malicious ratio $\beta$, $P(\xi_i=1)=\beta/(1+\beta)$ and $P(\xi_i=0)=1/(1+\beta)$. Substituting into \eqref{eq:map} and using conditional independence,
\begin{equation}\label{eq:mapfull}
	\begin{aligned}
		\theta^*=\arg\max_{\Theta}\prod_{i=1}^{N}\Bigg[
		&\frac{\beta}{1+\beta}\prod_{t=1}^{T}(1-\delta)^{q_i(t)}\delta^{1-q_i(t)}\\
		+&\frac{1}{1+\beta}\prod_{t=1}^{T}(1-\gamma^t)^{q_i(t)}(\gamma^t)^{1-q_i(t)}
		\Bigg].
	\end{aligned}
\end{equation}

\subsection{Per-Node Heterogeneous Fusion}
Theorem \ref{thm:hetero} shows that the \emph{evolution} of $p_m^t$ requires only $\eb$, but the fusion likelihood is improved by using the individual $\varepsilon_i$. Each UAV feeds back its sensing-SINR estimate over the ISAC link, and the fusion center obtains $\varepsilon_i$ from \eqref{eq:epsmap}. The per-node error probabilities become
\begin{equation}\label{eq:gammadelta_i}
	\begin{aligned}
		\gamma_i^t&=(1-\varepsilon_i)p_m^t+\varepsilon_i(1-p_m^t),\\
		\delta_i&=\varepsilon_i(1-P_a)+(1-\varepsilon_i)P_a,
	\end{aligned}
\end{equation}
and the per-slot MAP decision replaces \eqref{eq:mapfull} with
\begin{equation}\label{eq:mapnode}
	\begin{aligned}
		\hat\theta^t=\arg\max_{\theta\in\{0,1\}}\sum_{i=1}^{N}\ln\Big[
		&\tfrac{\beta}{1+\beta}P(r_i^t|\theta,\xi_i\!=\!1)\\
		&+\tfrac{1}{1+\beta}P(r_i^t|\theta,\xi_i\!=\!0)\Big],
	\end{aligned}
\end{equation}
where the two conditional probabilities are evaluated with $\delta_i$ and $\gamma_i^t$, respectively. Equation \eqref{eq:mapnode}, together with the evolutionary tracking of $p_m^t$, is termed evolutionary-game-theoretic MAP (EGTM) fusion; its homogeneous ($\varepsilon_i\equiv\eb$) and per-node variants are denoted by EGTM-homog and EGTM-pernode, respectively.
\begin{figure*}[t]
	\centerline{\includegraphics[width=7.3in,keepaspectratio]{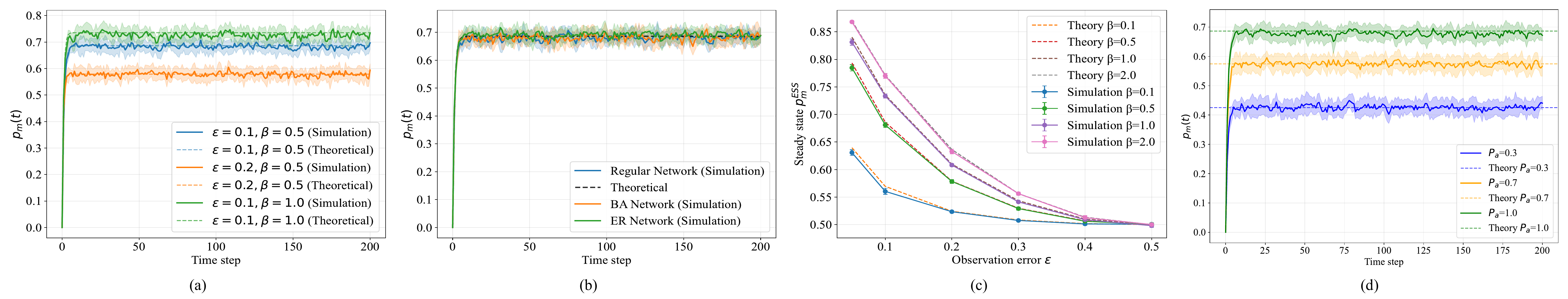}}
	\caption{Validation of the evolution dynamics and the steady-state landscape. (a) Evolution curves; (b) topology independence; (c) dependence of the ESS on $\varepsilon$ and $\beta$; (d) dependence of the ESS on $P_a$.}
	\label{fig2}
\end{figure*}
\subsection{Online Implementation}
Starting from the initial value $p_m^0=0$, the fusion center predicts $p_m^t$ online by propagating the evolution equation \eqref{eq:zeroth} (or \eqref{eq:firstorder} when $\alpha$ is not negligible), requiring only $(\eb,\beta,P_a,\alpha)$ and the payoff matrix. Directly evaluating \eqref{eq:mapfull} faces two problems: $2^T$ enumeration complexity and numerical underflow. This paper adopts a sliding window of length $\tau$: within each window $p_m^t$ is treated as constant, reducing the enumeration complexity to $2^\tau$, and all products are computed in the log domain. Since $\gamma^t$ is predicted slot by slot, the minimum window $\tau=1$ already achieves full accuracy, thereby yielding the lowest fusion latency, which directly supports the low-latency requirement of ISAC-6G situation inference. Algorithm \ref{alg:egtm} summarizes the per-node procedure.

\begin{algorithm}[t]
	\caption{EGTM per-node fusion (per slot)}
	\label{alg:egtm}
	\begin{algorithmic}[1]
		\REQUIRE Reports $\{r_i^t\}$, errors $\{\varepsilon_i\}$, parameters $(\beta,P_a)$, dynamics \eqref{eq:zeroth}/\eqref{eq:firstorder}
		\STATE Initialize $p_m^0=0$; compute $\delta_i$ via \eqref{eq:gammadelta_i}
		\FOR{each slot $t$}
		\STATE Propagate $p_m^{t}$ with the evolution equation
		\STATE Compute $\gamma_i^t$ for all $i$ via \eqref{eq:gammadelta_i}
		\STATE Accumulate the log-likelihood $\ell_\theta=\sum_i\ln\left[\frac{\beta}{1+\beta}P(r_i^t|\theta,\mathrm{mal})+\frac{1}{1+\beta}P(r_i^t|\theta,\mathrm{nor})\right]$ for $\theta\in\{0,1\}$
		\STATE Output $\hat\theta^t=\arg\max_\theta\ell_\theta$
		\ENDFOR
	\end{algorithmic}
\end{algorithm}

\section{Simulations and Analysis}\label{Sec7}
Monte Carlo simulations validate the dynamics, the ESS, and the fusion mechanism. Unless otherwise specified: $N=500$ UAVs on a random regular graph of degree $k=10$; payoff matrix $u_{ns}=0.8$, $u_{nd}=0.6$, $u_{ms}=0.6$, $u_{md}=0.4$; initial misinformation ratio $0$; evolution horizon $T=200$ ($300$ for the weak-selection experiments) to ensure stationarity; and fusion window $\tau=6$ except in the window-length experiment. To verify topological robustness, BA scale-free networks (average degree about 6) and Erd\H{o}s--R\'enyi networks ($p=k/N=0.02$) are also used. All experiments are independently repeated 40 times. The four sensing-error distributions in the heterogeneity experiments are: homogeneous ($\varepsilon_i\equiv\eb$); two-point (taking $\eb\pm0.05$ with equal probability); uniform ($\eb\pm0.08$); and SNR-based, whose log-normal SINR ($\sigma=5$~dB) is calibrated via \eqref{eq:epsmap} so that $\mathbb{E}[\varepsilon_i]=\eb$.
\begin{figure}[t]
	\centerline{\includegraphics[width=3.0in,keepaspectratio]{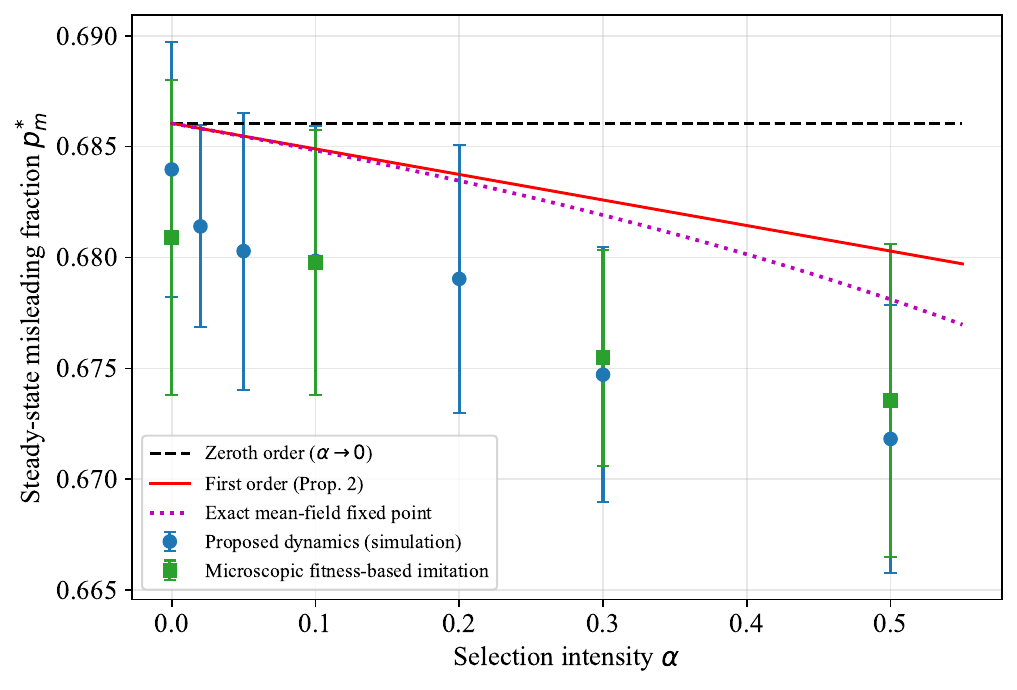}}
	\caption{Steady-state misinformation ratio versus selection intensity $\alpha$.}
	\label{fig11}
\end{figure}

\subsection{Evolution Dynamics}
Fig. \ref{fig2} summarizes the evolution dynamics and the steady-state landscape. (a) With the regular topology fixed, three parameter settings are compared: under the baseline ($\varepsilon=0.1$, $\beta=0.5$, $P_a=1$), $p_m(t)$ rises rapidly from $0$ and saturates at the theoretical ESS; increasing $\varepsilon$ to $0.2$ lowers the steady state (sensing noise dilutes the inducement), while increasing $\beta$ to $1.0$ raises it (more malicious neighbors amplify the attack). (b) With $\varepsilon=0.1$, $\beta=0.5$, $P_a=1$ fixed, the regular, BA, and ER topologies are compared; the three evolution curves almost coincide and agree with the theoretical ODE, confirming that the macroscopic dynamics is insensitive to the topology. (c) $\pmess$ decreases with $\varepsilon$ and increases with $\beta$, the two having opposite effects; at $\varepsilon=0.5$ the neighbor-influence term vanishes and all curves collapse to $0.5$. (d) Larger $P_a$ yields a higher steady state, consistent with the linear dependence in Corollary \ref{cor:flip}.

\subsection{First-Order Weak-Selection Effect}
Fig. \ref{fig11} validates Proposition \ref{prop:firstorder} ($\varepsilon=0.1$, $\beta=0.5$, $P_a=1.0$). The zeroth-order ESS $p_0^*=0.6860$, the first-order line \eqref{eq:ess1} ($G=-6.60\times10^{-3}$), and the exact mean-field fixed point \eqref{eq:Pexact} coincide as $\alpha\to0$. Superimposed are the proposed dynamics (exact transition probability with the realized local $\rho_i$) and microscopic fitness imitation. As $\alpha$ grows, the exact mean-field value drops from $0.6860$ to $0.6781$ ($\alpha=0.5$), and the first-order line is its tangent at the origin (deviation $<3\times10^{-5}$ at $\alpha=0.05$), validating the closed form $G$. Both microscopic simulations reproduce the monotone suppression; the slightly steeper descent stems from the graph-structural amplification effect \cite{DB,OhtsukiNature}. Thus payoff-driven conformity, rather than aiding the attack, mildly suppresses swarm-level misinformation under a consensus-oriented payoff regime.

\subsection{Robustness to Heterogeneous Sensing Errors}
Fig. \ref{fig12} validates Theorem \ref{thm:hetero}, with $N=300$, $\eb=0.1$, $P_a=1.0$, $\beta\in\{0.3,0.5,0.7\}$, and $T=300$. For each $\beta$, the population mean is fixed at $\eb=0.1$ and four error distributions with markedly different shapes are adopted. The inset compares the histograms of the four distributions. The bar chart shows that, for the same $\beta$, the steady-state means under the four distributions are almost equal, with cross-distribution dispersion below $0.002$, and all agree with the homogeneous theoretical value at $\varepsilon=\eb$ within the error bars. Even though the SNR distribution is clearly right-skewed (red in the inset), its ESS is governed solely by $\eb$, matching the prediction of Theorem \ref{thm:hetero}. This verifies that the fusion center needs no full error distribution and can predict the evolution trend of the misinformation ratio using only the population mean.
\begin{figure}[t]
	\centerline{\includegraphics[width=3.0in,keepaspectratio]{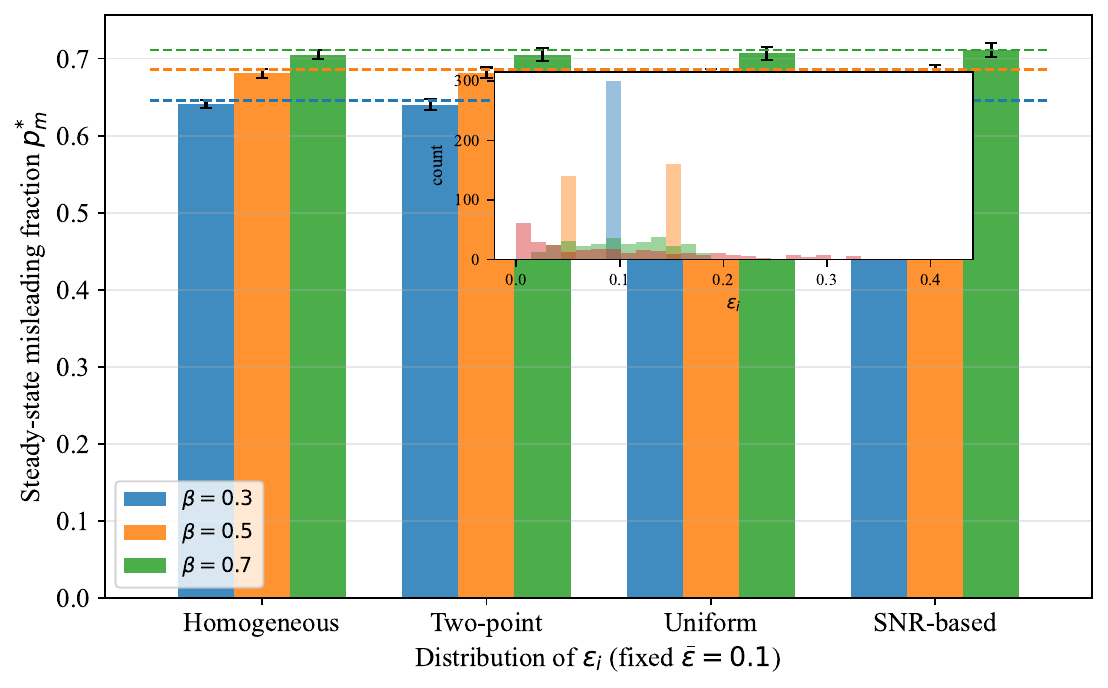}}
	\caption{Steady-state misinformation ratio under four distributions of $\varepsilon_i$ with the mean fixed at $\eb=0.1$. Dashed line: homogeneous theory \eqref{eq:ess}. Inset: the four $\varepsilon_i$ distributions.}
	\label{fig12}
\end{figure}

\subsection{Fusion Performance}
Fig. \ref{fig6} gives the accuracy of the MAP mechanism over an $(\varepsilon,\beta)$ grid with $P_a\in\{0.3,0.7,1.0\}$. Accuracy decreases with $\varepsilon$ and approaches the random-guess level $0.5$ as $\varepsilon=0.5$. For fixed $\varepsilon$, a larger $\beta$ instead \emph{improves} accuracy, because the malicious error probability $\delta$ is then more sharply characterized, enabling the fusion center to better separate malicious from normal reports; $P_a=1$ is the easiest to identify, since an adversary that always flips is statistically the most distinctive.

\begin{figure}[t]
	\centerline{\includegraphics[width=3.6in,keepaspectratio]{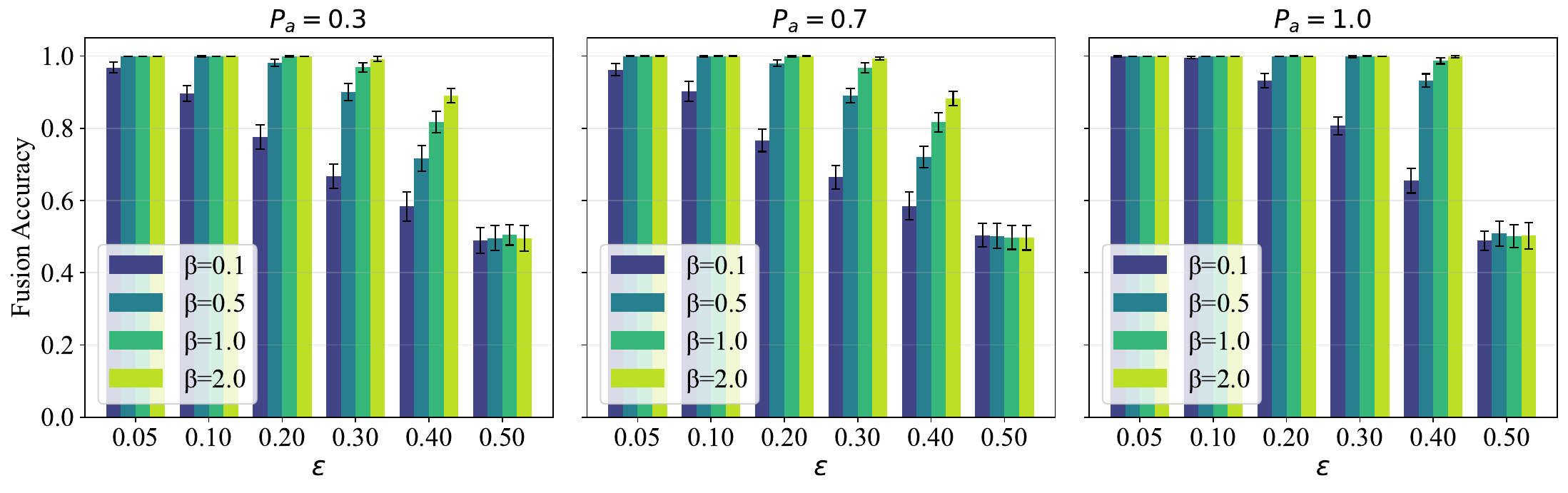}}
	\caption{Fusion accuracy of the proposed mechanism under different $(\varepsilon,\beta,P_a)$.}
	\label{fig6}
\end{figure}

Table \ref{tab:results} compares the fusion methods at $\varepsilon=0.1$ under a strong attack ($P_a=0.8$) and a weak attack ($P_a=0.3$). The baselines include: majority voting (MV); reputation-weighted fusion (RW) \cite{WSPRT}, whose reputation tracks consistency with the global decision; Beta-distribution reputation fusion (Beta) \cite{BetaCSS}; and independent fusion (IFM), which fixes the normal-node error probability at $\varepsilon$. Under the weak attack, Corollary \ref{cor:flip} gives $\pmess<1/2$, the majority is not systematically inverted, and all baselines perform comparably to the proposed method ($>0.99$). Under the strong attack, $\pmess$ crosses the majority-flip threshold and swarm-level misinformation occurs. MV follows the contaminated majority and systematically decides the wrong way (accuracy approaching $0$ rather than $0.5$); RW and Beta update reputations using the contaminated global decision as the reference, persistently rewarding misleading nodes and punishing honest ones; IFM suffers model mismatch due to its fixed error probability. Only the proposed mechanism, which predicts $\gamma^t$ online via the evolution equation, maintains accuracy $\ge 0.9995$ under all parameter combinations. This confirms that the failure boundary of existing defenses is exactly the majority-flip point $\pmess=1/2$ predicted by Corollary \ref{cor:flip}.

\begin{table}[t]
	\centering
	\caption{Accuracy comparison of fusion methods, $\varepsilon=0.1$}
	\label{tab:results}
	\begin{tabular}{lcccc}
		\toprule
		\multicolumn{5}{c}{$P_a=0.8$ (strong attack)}\\
		\midrule
		Method & $\beta=0.3$ & $\beta=0.5$ & $\beta=0.7$ & $\beta=0.9$ \\
		\midrule
		Proposed & $0.9997$ & $0.9995$ & $1.0000$ & $1.0000$ \\
		RW \cite{WSPRT} & $0.0029$ & $0.0016$ & $0.0009$ & $0.0009$ \\
		Beta \cite{BetaCSS} & $0.0037$ & $0.0021$ & $0.0010$ & $0.0009$ \\
		MV & $0.0027$ & $0.0013$ & $0.0009$ & $0.0009$ \\
		IFM & $0.0027$ & $0.0013$ & $0.0009$ & $0.0009$ \\
		Random & $0.5000$ & $0.4998$ & $0.5001$ & $0.4999$ \\
		\midrule
		\multicolumn{5}{c}{$P_a=0.3$ (weak attack)}\\
		\midrule
		Method & $\beta=0.3$ & $\beta=0.5$ & $\beta=0.7$ & $\beta=0.9$ \\
		\midrule
		Proposed & $0.9943$ & $0.9999$ & $1.0000$ & $1.0000$ \\
		RW \cite{WSPRT} & $0.9965$ & $0.9999$ & $1.0000$ & $1.0000$ \\
		Beta \cite{BetaCSS} & $0.9970$ & $1.0000$ & $1.0000$ & $1.0000$ \\
		MV & $0.9941$ & $0.9999$ & $1.0000$ & $1.0000$ \\
		IFM & $0.9942$ & $0.9999$ & $1.0000$ & $1.0000$ \\
		Random & $0.5038$ & $0.4992$ & $0.4998$ & $0.4991$ \\
		\bottomrule
	\end{tabular}
\end{table}
\begin{figure}[t]
	\centerline{\includegraphics[width=3.0in,keepaspectratio]{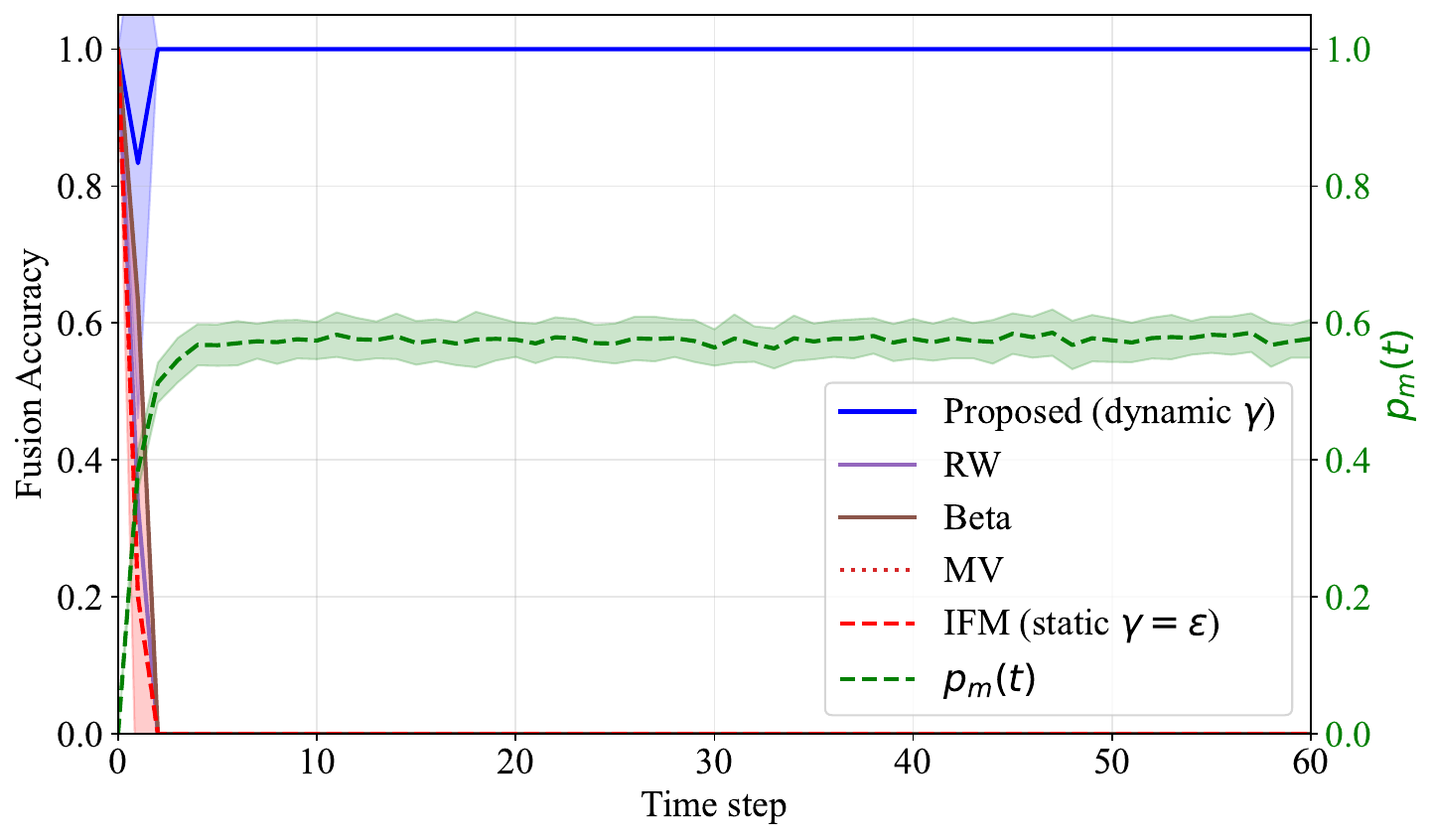}}
	\caption{Fusion accuracy versus evolution time ($\varepsilon=0.1$, $\beta=0.5$, $P_a=0.8$, per-slot decision).}
	\label{fig7}
\end{figure}

Fig. \ref{fig7} dissects the collapse process under a strong attack ($\varepsilon=0.1$, $\beta=0.5$, $P_a=0.8$) along the time dimension. In the initial stage the misinformation has not yet spread, and all methods decide correctly. Within about two slots, however, $p_m(t)$ crosses the majority-flip threshold $0.5$ (steady state about $0.58$), and the accuracies of MV, RW, Beta, and IFM collapse to zero simultaneously. Swarm-level misinformation breaks through the existing fusion mechanisms within two slots, and the reputation-based methods cannot recover because their reference (the global decision) is itself contaminated. The proposed method tracks $\gamma^t$ slot by slot and maintains accuracy close to $1.0$ throughout. The collapse time corresponds exactly to the threshold-crossing time, again validating Corollary \ref{cor:flip}.

\subsection{Fusion under Heterogeneous Errors}
Fig. \ref{fig13} scans the malicious ratio $\beta\in[0,1]$ under homogeneous and SNR-based heterogeneous errors, comparing EGTM-pernode (proposed), EGTM-homog, IFM, and majority voting. Two regimes stand out. (i) At $\beta=0$, although no attacker is present, imitation-based conformity alone drives the swarm to $\pmess=1/2$ (setting $\beta=0$ in \eqref{eq:ess} always yields $1/2$), so the reports lose their information content, and MV and IFM degenerate to the random level $\approx0.5$; the proposed mechanism can still extract the residual statistical asymmetry and maintains accuracy $0.83$--$0.85$. This self-inflicted effect is an intrinsic property of the DB update model in this paper, not a universal defect of consensus-based cooperative sensing \cite{QCons}. (ii) For any $\beta>0$, the baselines collapse across the board: over $\beta\in[0.05,0.5]$, IFM and MV drop to $0$--$0.25$ and do not recover reliably as $\beta$ grows further. EGTM-pernode maintains $0.83$--$1.0$ over the entire scan and reaches full accuracy from $\beta=0.2$ (homogeneous) or $\beta=0.3$ (SNR); EGTM-homog is slightly inferior under SNR heterogeneity, quantifying the gain of per-node error knowledge. These results confirm that the proposed framework is robust to both attack intensity and sensing heterogeneity.

\begin{figure}[t]
	\centerline{\includegraphics[width=3.6in,keepaspectratio]{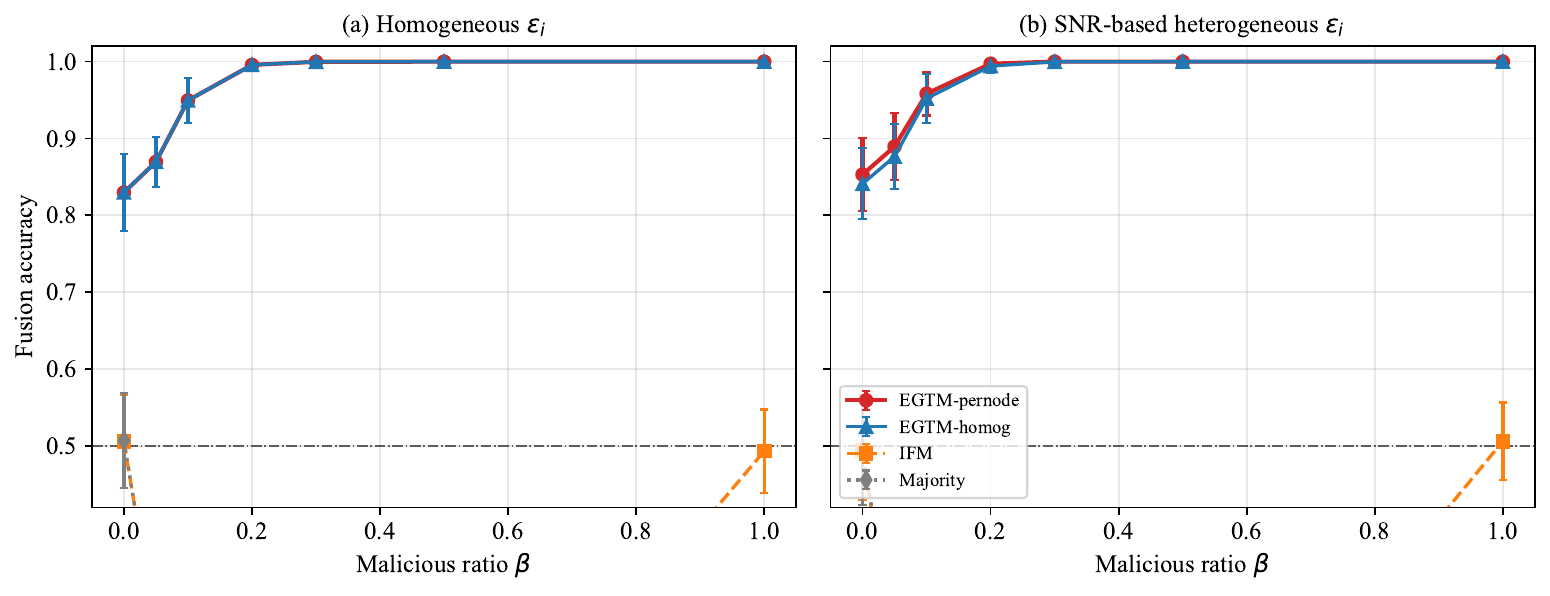}}
	\caption{Fusion accuracy versus malicious ratio $\beta$ under (a) homogeneous and (b) SNR-based heterogeneous sensing errors ($\eb=0.1$, $P_a=1.0$).}
	\label{fig13}
\end{figure}
\begin{figure*}[htbp]
	\centerline{\includegraphics[width=7.5in,keepaspectratio]{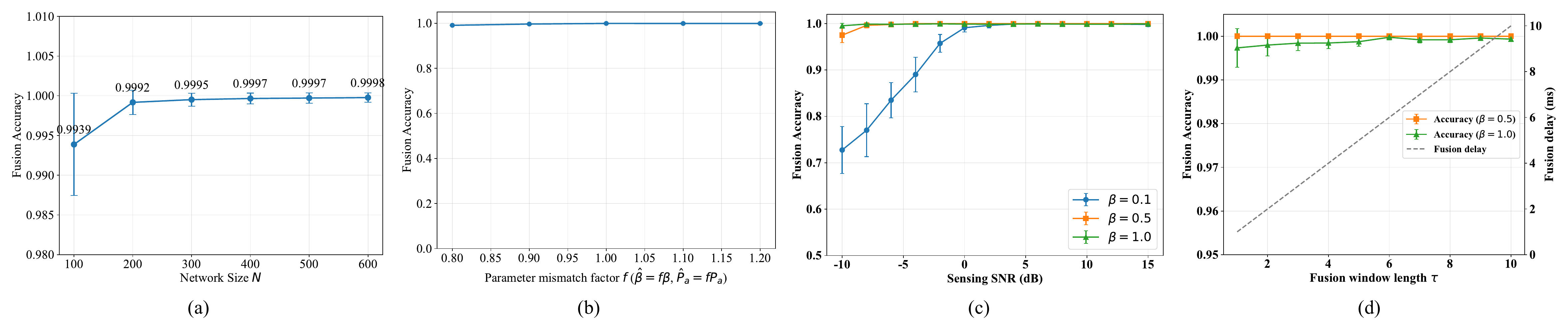}}
	\caption{Robustness and parameter sensitivity of the fusion performance. (a) Swarm size $N$; (b) parameter mismatch; (c) sensing SINR; (d) sliding-window length $\tau$.}
	\label{fig8}
\end{figure*}

\subsection{Robustness Analysis}
Fig. \ref{fig8} examines the robustness of the proposed fusion mechanism with respect to four classes of parameters.
Fig. \ref{fig8}(a) gives the accuracy of the proposed method under different $N$. Accuracy is close to $1.0$ at all scales and improves slightly with $N$, with reduced variance, reflecting the stabilizing effect of the law of large numbers; the method is applicable from small to large swarms.

The preceding experiments assume that the model parameters are known. Now let the fusion center use mismatched estimates $\hat\beta=f\beta$ and $\hat P_a=fP_a$ ($f\in[0.8,1.2]$), with $p_m^t$ predicted entirely by the theoretical evolution equation \eqref{eq:zeroth}. Fig. \ref{fig8}(b) shows that even under $\pm20\%$ mismatch the accuracy remains above $0.99$ ($0.9908$ at $f=0.8$); without mismatch, the ODE-prediction-based result ($0.9992$) is almost identical to the one using the true simulated $p_m^t$ ($0.9995$). The mechanism therefore relies on the dynamics model's characterization of how errors evolve, rather than on any omniscience assumption about the parameters.

Under additive white Gaussian noise, the bit error rate of binary detection is $\varepsilon=Q(\sqrt{2\gamma})$, which naturally maps the ISAC physical-layer SINR to the per-node sensing error. Fig. \ref{fig8}(c) plots the fusion accuracy against the SNR accordingly. Accuracy improves with the SNR and approaches $100\%$ above about $0$~dB for all $\beta$; in the low-SNR regime, a larger $\beta$ slightly improves accuracy because the known attack statistics of the malicious nodes become more distinctive. The proposed mechanism thus adapts to the wide SINR range of ISAC sensing channels.

Fig. \ref{fig8}(d) examines the window length $\tau$. The fusion latency grows linearly with $\tau$, while accuracy is almost insensitive to it: even at $\tau=1$, accuracy already reaches $\ge 0.99$ for $\beta\in\{0.5,1.0\}$, because the per-slot MAP decision already uses the online-predicted $\gamma^t$ and needs no long-window accumulation. The minimum window therefore yields the lowest latency without sacrificing accuracy, directly supporting low-latency situation inference in ISAC-6G swarms.

\section{Conclusion}\label{Sec8}
This paper investigated a critical vulnerability of imitation-based conformity cooperative sensing under Byzantine attacks in ISAC-enabled UAV swarms. Under the evolutionary-game framework, the conformity-aware evolution dynamics and ESS of the misinformation ratio were derived, and three theoretical results were established. First, under heterogeneous sensing errors, the per-node errors affect the zeroth-order ESS only through their mean, regardless of the shape of the distribution. Second, the first-order weak-selection closed-form correction shows that payoff-driven selection mildly suppresses swarm-level misinformation, and the accompanying exact mean-field fixed point holds for arbitrary selection intensity. Third, under DB imitation updating and a fixed attack intensity, swarm-level misinformation overwhelms the majority if and only if $P_a>1/2$, and this threshold is independent of both the sensing error and the malicious ratio. Embedding the predicted error dynamics into the MAP rule, the resulting fusion mechanism achieves nearly $100\%$ situation-inference accuracy under different topologies, attack intensities, network scales, and error distributions, and maintains accuracy above $99\%$ under $\pm20\%$ parameter mismatch; in contrast, majority voting, reputation weighting, and independent fusion collapse completely once the majority-flip threshold is crossed.

\bibliographystyle{IEEEtran}
\bibliography{MyRefs}

@article{AirGuard,
  author  = {Luo, H. and Chu, Z. and Zhang, T. and Zhao, C. and Lin, B. and Gao, F.},
  title   = {{AirGuard: UAV and bird recognition scheme for integrated sensing and communications system}},
  journal = {IEEE J. Sel. Areas Commun.},
  volume  = {44},
  pages   = {835--848},
  year    = {2025},
  doi     = {10.1109/JSAC.2025.3608760}
}

@article{RenPC,
title = {Integrated sensing and communication for low-altitude security},
journal = {Phys. Commun.},
volume = {76},
pages = {103117},
year = {2026},
issn = {1874-4907},
author = {Ruixing Ren and Junhui Zhao and Qingmiao Zhang}
}

@techreport{ITU2030,
  author      = {{International Telecommunication Union}},
  title       = {Framework and overall objectives of the future development of IMT for 2030 and beyond},
  institution = {ITU},
  type        = {Recommendation ITU-R M.2160-0},
  month       = nov,
  year        = {2023}
}

@article{ISACSurvey,
  author  = {Liu, F. and Cui, Y. and Masouros, C. and Xu, J. and Han, T. X. and Eldar, Y. C. and Buzzi, S.},
  title   = {{Integrated sensing and communications: Towards dual-functional wireless networks for 6G and beyond}},
  journal = {IEEE J. Sel. Areas Commun.},
  volume  = {40},
  number  = {6},
  pages   = {1728--1767},
  month   = jun,
  year    = {2022}
}

@article{Yao,
  author  = {Yao, Y. and Zhao, J. and Li, Z. and Cheng, X. and Wu, L.},
  title   = {{Jamming and eavesdropping defense scheme based on deep reinforcement learning in autonomous vehicle networks}},
  journal = {IEEE Trans. Inf. Forensics Security},
  volume  = {18},
  pages   = {1211--1224},
  year    = {2023},
  doi     = {10.1109/TIFS.2023.3236788}
}

@ARTICLE{RenTVTIoV,
  author={Zhao, Junhui and Ren, Ruixing and Zou, Dan and Zhang, Qingmiao and Xu, Wei},
  journal={IEEE Trans. Veh. Technol.}, 
  title={{IoV-Oriented Integrated Sensing, Computation, and Communication: System Design and Resource Allocation}}, 
  year={2024},
  month={Nov.},
  volume={73},
  number={11},
  pages={16283-16294}
}

@article{RenUAV,
  author  = {Ren, R. and Zhao, J. and Zhang, Q.},
  title   = {{UAV-assisted collaborative sensing task offloading and resource allocation in IoV}},
  journal = {IEEE Trans. Veh. Technol.},
  volume  = {75},
  number  = {4},
  pages   = {6806--6815},
  year    = {2026},
  doi     = {10.1109/TVT.2025.3623590}
}

@article{RenITS,
  author  = {Ren, R. and Zhao, J. and Zou, D. and Zhang, Q. and Wang, D. and Xu, W.},
  title   = {Collaborative computation in integrated sensing, communication, and computation system for autonomous driving},
  journal = {IEEE Trans. Intell. Transp. Syst.},
  volume  = {27},
  number  = {1},
  pages   = {883--894},
  year    = {2026},
  doi     = {10.1109/TITS.2025.3625173}
}

@article{RenIOTJ,
  author  = {Ren, R. and Zhao, J. and Zhang, Q. and Wang, D. and Li, J.},
  title   = {{DRL beamforming in RIS-aided IoV for integrated-sensing-communication-computation}},
  journal = {IEEE Internet Things J.},
  volume  = {12},
  pages   = {28201--28213},
  year    = {2025},
  doi     = {10.1109/JIOT.2025.3566380}
}

@article{Vicsek,
  author  = {Vicsek, T. and Czir{\'o}k, A. and Ben-Jacob, E. and Cohen, I. and Shochet, O.},
  title   = {Novel type of phase transition in a system of self-driven particles},
  journal = {Phys. Rev. Lett.},
  volume  = {75},
  number  = {6},
  pages   = {1226--1229},
  month   = aug,
  year    = {1995}
}

@article{MDCons,
  author  = {Xu, P. and Sun, X. and Yu, M. and Liu, J. and Bai, T.},
  title   = {{A consensus control method for unmanned aerial vehicle (UAV) swarm based on molecular dynamics}},
  journal = {Drones},
  volume  = {9},
  number  = {6},
  pages   = {Art. no. 404},
  year    = {2025},
  doi     = {10.3390/drones9060404}
}

@article{QCons,
  author  = {Zhu, S. and Chen, B.},
  title   = {Distributed detection in ad hoc networks through quantized consensus},
  journal = {IEEE Trans. Inf. Theory},
  volume  = {64},
  number  = {11},
  pages   = {7017--7030},
  month   = nov,
  year    = {2018}
}

@article{CSS,
  author  = {Chouhan, A. and Captain, K. and Parmar, A. and Patel, J.},
  title   = {Defending cooperative spectrum sensing from Byzantine attacks: An effective entropy-based weighted algorithm},
  journal = {IEEE Wireless Commun. Lett.},
  volume  = {12},
  number  = {12},
  pages   = {2063--2067},
  month   = dec,
  year    = {2023}
}

@article{DFalse,
  author  = {Quan, C. and Bulusu, S. and Geng, B. and Han, Y. S. and Sriranga, N. and Varshney, P. K.},
  title   = {On ordered transmission based distributed Gaussian shift-in-mean detection under Byzantine attacks},
  journal = {IEEE Trans. Signal Process.},
  volume  = {71},
  pages   = {3343--3356},
  year    = {2023},
  doi     = {10.1109/TSP.2023.3314270}
}

@article{PrivCons,
  author  = {Zhu, C. and Jiang, B. and Zhang, B. and Xu, H. and Wu, Y. and Xu, M.},
  title   = {{Privacy-preserving average consensus for swarm systems subject to DoS attacks}},
  journal = {IEEE Internet Things J.},
  volume  = {12},
  number  = {24},
  pages   = {53659--53669},
  year    = {2025},
  doi     = {10.1109/JIOT.2025.3617342}
}

@inproceedings{SwarmSpoof,
  author    = {Yao, Y. and Dash, P. and Pattabiraman, K.},
  title     = {May the swarm be with you: Sensor spoofing attacks against drone swarms},
  booktitle = {Proc. ACM SIGSAC Conf. Comput. Commun. Security (CCS)},
  pages     = {3511--3513},
  year      = {2022},
  doi       = {10.1145/3548606.3563535}
}

@article{WSPRT,
  author  = {Xiao, S. and Wu, J. and Lin, P. and Qiao, L. and Qiu, Z. and Su, M.},
  title   = {Reputation-based self-differential sequential mechanism for collaborative spectrum sensing against Byzantine attack in cognitive wireless sensor networks},
  journal = {IEEE Sensors Lett.},
  volume  = {8},
  number  = {10},
  pages   = {Art. no. 7501104},
  month   = oct,
  year    = {2024},
  doi     = {10.1109/LSENS.2024.3454708}
}

@article{BetaCSS,
  author  = {Wu, J. and Gao, S. and Teng, X. and Zhang, Z. and Dai, M. and Ge, H. and Cao, W.},
  title   = {Beta distribution function-based cooperative spectrum sensing against Byzantine attack in cognitive wireless sensor networks},
  journal = {IEEE Sensors Lett.},
  year    = {2024}
}

@phdthesis{LinPhD,
  author = {Lin, Y.},
  title  = {Analysis of public opinion security based on bounded rational behaviors in social network},
  school = {Tsinghua University},
  address = {Beijing, China},
  year   = {2025}
}

@article{Centola,
  author  = {Centola, D.},
  title   = {The spread of behavior in an online social network experiment},
  journal = {Science},
  volume  = {329},
  number  = {5996},
  pages   = {1194--1197},
  year    = {2010}
}

@article{Kucherov,
  author  = {Kucherov, D. P. and Jiang, G. and Liu, H. and Fu, M.},
  title   = {{UAV group control protocol with adaptive consensus}},
  journal = {Int. J. Adapt. Control Signal Process.},
  volume  = {38},
  number  = {9},
  pages   = {3177--3194},
  year    = {2024},
  doi     = {10.1002/acs.3868}
}

@article{HanCons,
  author  = {Han, F. and Liu, J. and Li, J. and Song, J. and Wang, M. and Zhang, Y.},
  title   = {Consensus control for multi-rate multi-agent systems with fading measurements: The dynamic event-triggered case},
  journal = {Syst. Sci. Control Eng.},
  volume  = {11},
  number  = {1},
  pages   = {Art. no. 2158959},
  year    = {2023},
  doi     = {10.1080/21642583.2022.2158959}
}

@article{WenSwarm,
  author  = {Wen, L. and Zhen, Z. and Tao, C. and Ding, J.},
  title   = {{Distributed cooperative strategy of UAV swarm without speed measurement under saturation attack mission}},
  journal = {IEEE Trans. Aerosp. Electron. Syst.},
  volume  = {60},
  number  = {4},
  pages   = {4518--4529},
  year    = {2024},
  doi     = {10.1109/TAES.2024.3382627}
}

@article{ShangPred,
  author  = {Shang, P. and Peng, Z. and He, H. and Li, T. and Liu, G.},
  title   = {{Trajectory prediction of dynamic UAV swarm with interaction and quantity uncertainty under saturation attack mission}},
  journal = {IEEE Trans. Aerosp. Electron. Syst.},
  volume  = {62},
  pages   = {8136--8152},
  year    = {2026},
  doi     = {10.1109/TAES.2026.3674892}
}

@article{BlockchainTrust,
  author  = {Zhao, J. and Huang, F. and Liao, L. and Zhang, Q.},
  title   = {Blockchain-based trust management model for vehicular ad hoc networks},
  journal = {IEEE Internet Things J.},
  volume  = {11},
  number  = {5},
  pages   = {8118--8132},
  year    = {2024},
  doi     = {10.1109/JIOT.2023.3318597}
}

@article{SpoofDetect,
  author  = {Bi, S. and Li, K. and Hu, S. and Ni, W. and Wang, C. and Wang, X.},
  title   = {{Detection and mitigation of position spoofing attacks on cooperative UAV swarm formations}},
  journal = {IEEE Trans. Inf. Forensics Security},
  volume  = {19},
  pages   = {1883--1895},
  year    = {2024},
  doi     = {10.1109/TIFS.2023.3341398}
}

@article{CuiResilient,
  author  = {Cui, Y. and Liang, Y. and Luo, Q. and Shu, Z. and Huang, T.},
  title   = {{Resilient consensus control of heterogeneous multi-UAV systems with leader of unknown input against Byzantine attacks}},
  journal = {IEEE Trans. Autom. Sci. Eng.},
  volume  = {22},
  pages   = {5388--5399},
  year    = {2024},
  doi     = {10.1109/TASE.2024.3420697}
}

@article{GongHIL,
  author  = {Gong, X. and Gui, J. and Chen, Y. and Yang, X. and Yu, W. and Huang, T.},
  title   = {{Resilient human-in-the-loop formation-tracking of multi-UAV systems against Byzantine attacks}},
  journal = {IEEE Trans. Autom. Sci. Eng.},
  volume  = {22},
  pages   = {3797--3809},
  year    = {2024},
  doi     = {10.1109/TASE.2024.3400155}
}

@article{GongTwoLayer,
  author  = {Gong, X. and Basin, M. V. and Feng, Z. G. and Huang, T. W. and Cui, Y. K.},
  title   = {{Resilient time-varying formation-tracking of multi-UAV systems against composite attacks: A two-layered framework}},
  journal = {IEEE/CAA J. Autom. Sinica},
  volume  = {10},
  number  = {4},
  pages   = {969--984},
  year    = {2023},
  doi     = {10.1109/JAS.2023.123339}
}

@article{LiAirSea,
  author  = {Li, J. and Liu, J. J. R. and Liu, C. and Baldi, S. and Yu, D.},
  title   = {{Resilient control under DoS attacks of hybrid air--sea swarms with cooperative--competitive interactions}},
  journal = {IEEE Trans. Aerosp. Electron. Syst.},
  volume  = {62},
  pages   = {2913--2927},
  year    = {2026},
  doi     = {10.1109/TAES.2025.3643810}
}

@article{WuZeroSum,
  author  = {Wu, Y. and Chen, M. and Chadli, M.},
  title   = {{Zero-sum-game-based distributed fuzzy adaptive self-triggered control of swarm UAVs under intermittent communication and DoS attacks}},
  journal = {IEEE Trans. Fuzzy Syst.},
  volume  = {32},
  number  = {9},
  pages   = {5371--5384},
  year    = {2024},
  doi     = {10.1109/TFUZZ.2024.3423709}
}

@article{DB,
  author  = {Ohtsuki, H. and Nowak, M. A.},
  title   = {The replicator equation on graphs},
  journal = {J. Theor. Biol.},
  volume  = {243},
  number  = {1},
  pages   = {86--97},
  year    = {2006}
}

@article{OhtsukiNature,
  author  = {Ohtsuki, H. and Hauert, C. and Lieberman, E. and Nowak, M. A.},
  title   = {A simple rule for the evolution of cooperation on graphs and social networks},
  journal = {Nature},
  volume  = {441},
  pages   = {502--505},
  year    = {2006}
}

@article{NowakFive,
  author  = {Nowak, M. A.},
  title   = {Five rules for the evolution of cooperation},
  journal = {Science},
  volume  = {314},
  number  = {5805},
  pages   = {1560--1563},
  year    = {2006}
}
\end{document}